\documentclass[11pt,reqno,a4paper]{amsart}
\usepackage{amsfonts}
\usepackage{ifthen}
\usepackage{amsthm,cite}
\usepackage{amsmath,mathrsfs}
\usepackage{graphicx}
\usepackage{subcaption}
\usepackage{amscd,amssymb,amsthm}
\usepackage{color}
\usepackage{hyperref}
\usepackage{amsthm}
\usepackage{amsmath}

\allowdisplaybreaks

\usepackage{array,multirow,makecell}

\setcellgapes{1pt}
\makegapedcells
\newcolumntype{R}[1]{>{\raggedleft\arraybackslash }b{#1}}
\newcolumntype{L}[1]{>{\raggedright\arraybackslash }b{#1}}
\newcolumntype{C}[1]{>{\centering\arraybackslash }b{#1}}
\newcounter{minutes}
\divide\time by 60
\newcounter{hours}
\multiply\time by 60 \addtocounter{minutes}{-\time}

\newtheorem{theo}{Theorem}[section]

\newtheorem{theorem}{Theorem}[section]

\newtheorem{lemma}[theorem]{Lemma}

\newtheorem{corollary}[theorem]{Corollary}
\newtheorem{remark}[theorem]{Remark}

\usepackage{tikz,xcolor,hyperref}
\definecolor{lime}{HTML}{A6CE39}
\DeclareRobustCommand{\orcidicon}{%
	\begin{tikzpicture}
	\draw[lime, fill=lime] (0,0)
	circle [radius=0.16]
	node[white] {{\fontfamily{qag}\selectfont \tiny ID}};
	\draw[white, fill=white] (-0.0625,0.095)
	circle [radius=0.007];
	\end{tikzpicture}
	\hspace{-2mm}
}

\foreach \x in {A, ..., Z}{%
	\expandafter\xdef\csname orcid\x\endcsname{\noexpand\href{https://orcid.org/\csname orcidauthor\x\endcsname}{\noexpand\orcidicon}}
}

\title[]{A new class of coherent states involving Fox-Wright functions and their generalization in the bicomplex framework}
\author[Snehasis Bera]{Snehasis Bera}
\address{Snehasis Bera\newline Department of Mathematics,\newline National Institute of Technology Jamshedpur, Jamshedpur-831014, Jharkhand, India.}
\email{berasnehasis1996@gmail.com}
\author[Sourav Das]{Sourav Das}
\thanks{$^\ast$Corresponding author}
\address{Sourav Das\newline Department of Mathematics,\newline National Institute of Technology Jamshedpur, Jamshedpur-831014, Jharkhand, India.}
\email{souravdasmath@gmail.com, souravdas.math@nitjsr.ac.in}
\author[Abhijit Banerjee]{Abhijit Banerjee$^\ast$}
\address{Abhijit Banerjee\newline Department of Mathematics,\newline Garhbeta College, Paschim Medinipur-721127, West Bengal, India.}
\email{abhijit.banerjee.81@gmail.com}
\date{}
\begin{document}

\begin{abstract}
  In this work, an extensive class of coherent states is introduced by taking the Fox–Wright function as the normalization function. It is demonstrated that these states satisfy the key requirements of continuity, normalizability and resolution of unity. Furthermore, coherent states associated with the continuous spectrum are obtained through a discrete-to-continuous limiting procedure. Moreover, FW-generalized multi-parameter $\nu$-function is introduced and shown to act as the normalization function for the Fox–Wright coherent states in the continuous spectrum. Later the Fox–Wright function with bicomplex arguments has been introduced and its existence has been investigated. Bicomplex Fox–Wright coherent states are also developed for the discrete spectrum based on this new function and their properties are analyzed. Subsequently, the results regarding Fox-Wright coherent states are generalized to the bicomplex setting. In addition, a bicomplex FW-generalized multi-parameter $\nu$-function is defined to demonstrate that it provides the normalization for these states in the continuous spectrum. 
\end{abstract}

\maketitle

\section{Introduction and motivation}
   Exactly a century ago, E. Schr\"{o}dinger \cite{sch} introduced the problem of identifying quantum states that replicate the behavior of corresponding classical states. Later in 1963, R. J. Glauber \cite{coh_incoh} extended Schrödinger’s idea to quantum electrodynamics and formally introduced the term ‘coherent states’. Over the following decades, attention was devoted to the investigation of the properties and applications of coherent states for the one-dimensional harmonic oscillator. As these coherent states are defined in terms of the bosonic annihilation and creation operators $\hat{a}$ and $\hat{a}^\dagger$, which satisfy the canonical commutation relation
$[\hat{a}, \hat{a}^\dagger] = 1$, they are known as canonical coherent states. Simultaneously, it was recognized that coherent states are not confined to the one-dimensional harmonic oscillator, whose energy spectrum varies linearly with the principal quantum number $n$ rather, this formalism can be generalized to systems with nonlinear energy spectrum. For nonlinear deformed oscillators \cite{deformed_osc}, one can construct coherent states through the annihilation and creation operators
 \begin{equation}\label{2}
  A=Bf(N), A^\dag=f(N)B^\dag\nonumber
\end{equation}
which satisfy the deformed commutation relations
\begin{equation}\label{3}
[N,A]=-A, [N,A^\dag]=A^\dag \mbox{ and }[A,A^\dag]=(N+1)f^2(N+1)-Nf^2(N)\nonumber
\end{equation}
where $f$ denotes a hermitian operator-valued function of number operator and coherent states can be defined as the eigenstates of $A$. Today, coherent states are widely used in many areas of physics, such as signal and quantum information processing, quantum optics and beyond. Various extension of coherent states have been developed, including generalized hypergeometric coherent states \cite{gh_cs}, Mittag Lafler coherent states \cite{ml_cs}, coherent states for generalized Laguerre functions \cite{lg_cs}, generalized coherent states associated with $C_\lambda$-extended oscillator \cite{ann_phys} etc. Numerous important books and research articles have been written on coherent states and their various applications see \cite{coherent_app,deformed_osc,gen_coherent_app}. More interestingly, in recent time, bicomplex quantum mechanics has emerged as a viable extension of standard quantum theory (see \cite{BC-H,bc_hamiltonian} and references therein), and a few works extending the notion of coherent states into the bicomplex domain \cite{bicomplex ghf,bicomplex_bessel} has appeared in the literature.

Recently Popov \cite{int_mittage}, constructed and studied coherent states associated with the multi-index Mittag–Leffler function. Subsequently, by applying the discrete-to-continuous limit, coherent states for a continuous spectrum were derived and their properties were analyzed. Next, Popov \cite{applfox}, developed coherent states employing the Fox–Wright type function and examined their associated properties. Motivated by this, we construct Fox–Wright coherent states in which the Fox–Wright function acts as the normalization function and  investigate their properties. These states are distinct from the generalized coherent states reported in \cite{applfox}. Furthermore, we construct coherent states corresponding to the continuous spectrum and show that they fulfill the requirements of continuity, normalizability, and resolution of unity. Moreover, the Fox–Wright function with bicomplex arguments is introduced and its existence is investigated. Subsequently, all results concerning Fox–Wright coherent states are extended to the bicomplex setting.

The Fox–Wright function is an important special function that plays a significant role in fractional calculus, approximation theory, mathematical physics and related areas of science and engineering. Many special functions, including Wright function \cite{wright}, generalized hypergeometric function \cite{sp_function}, Mittag–Leffler function \cite{mittagleffler} and Bessel function \cite{bessel} can be represented as particular cases of the Fox–Wright function.
The Fox-Wright function ${}_p\psi_q[\cdot]$ generalization of the hypergeometric function is defined as follows \cite{asy_ghf,asymghf} :
    \begin{align}\label{eq:cfw1}
         {}_{p}\psi_{q}\left[\begin{matrix}&(a_1,A_1),\ldots &(a_p,A_p);\\& (b_1,B_1),\ldots&(b_q,B_q);&\end{matrix}z\right]=\sum_{n=0}^{\infty}\frac{\prod\limits_{l=1}^{p}\Gamma(a_l+nA_l)}{\prod\limits_{r=1}^{q}\Gamma(b_r+nB_r)}\frac{z^n}{n!},
    \end{align}
    where $a_l,b_r\in\mathbb{C}$ and $A_l,B_r\in\mathbb{R}^+$ for $l=1,2,\ldots p, r=1,2,\ldots q$. The series \eqref{eq:cfw1}, converges absolutely and uniformly in the whole complex space if 
    \begin{align}\label{eq:vt5}
        1+\sum_{r=1}^qB_r-\sum_{l=1}^pA_l>0.
    \end{align}

The structure of the paper is as follows: in section \ref{sec2}, Fox–Wright coherent states are constructed for both discrete and continuous spectrum and their properties are examined. In section \ref{sec3}, the Fox–Wright function in the bicomplex setting is defined and its existence is investigated. In section \ref{sec5}, all results related to Fox–Wright coherent states are extended to the bicomplex setting.  
\section{Construction of Fox-Wright coherent states}\label{sec2}
In this section, we define Fox-Wright coherent states, categorize them according to well-defined convergence criteria and identifies the parameter constraints governing these states. Moreover, we have derived Fox-Wright
coherent states corresponding to the continuous spectrum by applying the discrete-continuous limit. We define Fox-Wright states in the infinite dimensional Hilbert space of the Fock vectors $\mid k>, k=1,2,\ldots$ as follows
\begin{eqnarray}\label{eq:gh9}
\mid z> =\frac{1}{\sqrt{\mathcal{N}^{(p,q)}(\mid z\mid^2)}}\sum_{k=0}^\infty \frac{z^k}{\sqrt{\rho^{(p,q)}(k)}}\mid k>,
\end{eqnarray}
where the normalization functions are expressed through Fox–Wright functions as given below
\begin{align}\label{eq:a8f}
  \mathcal{N}^{(p,q)}(\zeta)&= \frac{\prod\limits_{r=1}^q \Gamma(b_r)}{\prod\limits_{l=1}^p\Gamma(a_l)} {}_{p}\psi_{q}\left[\begin{matrix}&(a_1,A_1),\ldots &(a_p,A_p);\\& (b_1,B_1),\ldots&(b_q,B_q);&\end{matrix}\zeta\right]\nonumber\\
&=\frac{\prod\limits_{r=1}^q \Gamma(b_r)}{\prod\limits_{l=1}^p\Gamma(a_l)}\sum_{k=0}^{\infty}\frac{\prod\limits_{l=1}^{p}\Gamma(a_l+kA_l)}{\prod\limits_{r=1}^{q}\Gamma(b_r+kB_r)}\frac{\zeta^k}{k!}
\end{align}
 for $\zeta=\mid z \mid^2$ and the parameter function $\rho^{(p,q)}(k)$ in (\ref{eq:gh9}) is defined by
\begin{eqnarray}\label{eq:vt9}
\rho^{(p,q)}(k) = \Gamma(k+1)\frac{\prod\limits_{l=1}^p\Gamma(a_l)}{\prod\limits_{r=1}^q \Gamma(b_r)}\frac{\prod\limits_{r=1}^{q}\Gamma(b_r+kB_r)}{\prod\limits_{l=1}^{p}\Gamma(a_l+kA_l)}.
\end{eqnarray}
Let us consider complex valued function
\begin{equation*}
f^{(p,q)}(s)=\sqrt{\frac{\prod\limits_{l=1}^p\Gamma(a_l+sA_l)\prod\limits_{r=1}^q\Gamma(b_r+(s+1)B_r)}{\prod\limits_{l=1}^p\Gamma(a_l+(s+1)A_l)\prod\limits_{r=1}^q\Gamma(b_r+sB_r)}(s+1)},
\end{equation*}
such that
\begin{equation}\label{eq:vt2}
\prod_{s=0}^{k-1}f^{(p,q)}(s)=\sqrt{\frac{\prod\limits_{l=1}^p\Gamma(a_l)\prod\limits_{r=1}^q\Gamma(b_r+kB_r)}{\prod\limits_{l=1}^p\Gamma(a_l+kA_l)\prod\limits_{r=1}^q\Gamma(b_r)}\Gamma(k+1)}=\sqrt{\rho^{(p,q)}(k)}.
\end{equation}
From \eqref{eq:vt9} and \eqref{eq:vt2}, we obtain the following recurrence relations
\begin{equation}\label{eq:cza9}
\rho^{(p,q)}(0)=1\;\; \mbox{and}\;\;\rho^{(p,q)}(k+1)=\rho^{(p,q)}(k)\left(f^{(p,q)}(k)\right)^2.
\end{equation}
 From the above construction, it follows that $\rho^{(m,n)}(k)$ must be a strictly positive real number. Consequently, this leads to restrictions on the parameters 
\begin{equation*}
a_l>0,b_r>0\;\quad \mbox{for}\;l=1,2,...,p; r=1,2,...,q.
\end{equation*}
Now using the normalization functions, we can compute the scalar product between two Fox–Wright states as follows:
\begin{equation}\label{eq:xd6}
<z \mid z'>=\frac{\mathcal{N}^{(p,q)}(z*z')}{\sqrt{\mathcal{N}^{(p,q)}(\mid z\mid^2)}\sqrt{\mathcal{N}^{(p,q)}(\mid z'\mid)}}.
\end{equation}
Form expression \eqref{eq:xd6}, it follows that the Fox–Wright states are normalized but not orthogonal and this expression is well defined provided that the Fox–Wright functions appearing in the scalar product are convergent. The convergence condition for the Fox–Wright function is given in \eqref{eq:vt5} namely
$$\sum_{r=1}^qB_r-\sum_{l=1}^pA_l>-1.$$
Now, we introduce annihilation and creation operators as follows
\begin{eqnarray}\label{eq:cz7}
\mathcal{A}_-&=&\sum_{k=0}^\infty f^{(p,q)}(k)\mid k><k+1\mid\nonumber\\
\mathcal{A}_+&=&\sum_{k=0}^\infty f^{(p,q)}(k)\mid k+1><k\mid,\nonumber
\end{eqnarray}
where $\mathcal{A}_+ = \mathcal{A}_-^\dagger$ that is $\mathcal{A}_+$ represents the adjoint operator corresponding to $\mathcal{A}_-$ and satisfies the following relations:
\begin{align}\label{eq:csn5}
&\mathcal{A}_-\mid k>=f^{(p,q)}(k-1)\mid k-1>,\nonumber\\
&\mathcal{A}_+\mid k>=f^{(p,q)}(k)\mid k+1>,\\
&<k\mid \mathcal{A}_-\mathcal{A}_+\mid k>=\left[f^{(m,n)}(k)\right]^2,\nonumber
\end{align}
and
\begin{equation}\label{eq:csf6}
<k\mid \mathcal{A}_+\mathcal{A}_-\mid k>=\left[f^{(m,n)}(k-1)\right]^2.
\end{equation}
The operators $\mathcal{A}_-$ and $\mathcal{A}_+$ do not commute and their noncommutative commutation relation is given by:
\begin{equation}\label{eq:csn9}
\left[\mathcal{A}_-,\mathcal{A}_+\right]=\sum_{k=0}^{\infty}\left(\left[f^{(m,n)}(k)\right]^2-\left[f^{(m,n)}(k-1)\right]^2\right)\mid k><k\mid.\nonumber
\end{equation}
Using the recurrence relation (\ref{eq:cza9}), we obtain
\begin{eqnarray}
\mathcal{A}_-\mid z>&=&\left(\sum_{k=0}^\infty f^{(p,q)}(k)\mid k><k+1\mid\right)\frac{1}{\sqrt{\mathcal{N}^{(p,q)}(\mid z\mid^2)}}\left(\sum_{k=0}^\infty \frac{z^n}{\sqrt{\rho^{(p,q)}(k)}}\mid k>\right)\nonumber\\
&=&\sum_{k=0}^\infty \frac{f^{(p,q)}(k)}{\sqrt{\mathcal{N}^{(p,q)}(\mid z\mid^2)}}\frac{z^{k+1}}{\sqrt{\rho^{(p,q)}(k+1)}}\mid k>\nonumber\\
&=&z\mid z>.\nonumber
\end{eqnarray}
 Thus, the Fox–Wright states $\mid z>$ defined in \eqref{eq:gh9} are eigenstates of the annihilation operator $\mathcal{A}_-$ corresponding eigenvalue $z$ confirming that they are coherent states. Next, we verify that Fox-Wright coherent states satisfy the following fundamental properties: 
 \begin{itemize}
    \item [(a)] Continuity: $|z'-z|\rightarrow 0$ implies that $||\mid z'>-\mid z>||\rightarrow 0$.
    The continuity property is satisfied since the series appearing in the definition of the coherent states is a continuous function of the complex variable $z$ over the entire complex plane.
    \item [(b)] Normalization: $<z \mid z>=1$. This property is a direct consequence of the previously derived expression \eqref{eq:xd6} for the overlap of two coherent states.
    \item[(c)] Resolution of unity: A positive weight function $W(|z|^2)$, is chosen of the integration measure 
$$d\mu(z)=\frac{d\theta}{2\pi}d(|z|^2)W(|z|^2),\;\;0<\theta<2\pi$$ such that the integral $\int d \mu(z)\mid z><z\mid=1$ holds. \end{itemize}
Let us take $\mid 0>$ denote the ground state satisfying $\mathcal{A}_- \mid 0>=0\mid 0>$. Now using \eqref{eq:cz7}, we obtain
\begin{align*}
    (\mathcal{A}_+)^k\mid 0>=\prod_{s=0}^{k-1}f(s)\mid k>
    =\sqrt{\rho^{(p,q)}(k)}\mid k>.
\end{align*}
Since, $\mathcal{A}_+$ represents the adjoint operator corresponding to $\mathcal{A}_-$, the relations below hold
\begin{align}\label{eq:csb7}
    \mid k>=\frac{1}{\sqrt{\rho^{(p,q)}(k)}}(\mathcal{A}_+)^k\mid 0>\;\; \mbox{and} \;\;< k \mid=\frac{1}{\sqrt{\rho^{(p,q)}(k)}} <0 \mid (\mathcal{A}_-)^k.
\end{align}
Using \eqref{eq:csb7} and \eqref{eq:gh9}, we obtain
\begin{align}\label{eq:z9x}
    \mid z>= \frac{1}{\sqrt{\mathcal{N}^{(p,q)}(\mid z\mid^2)}}\sum_{k=0}^\infty \frac{(z\mathcal{A}_+)^k}{{\rho^{(p,q)}(k)}}\mid 0>\;\;\mbox{and}\;\;<z \mid= \frac{1}{\sqrt{\mathcal{N}^{(p,q)}(\mid z\mid^2)}}<0 \mid\sum_{k=0}^\infty \frac{(z^\ast\mathcal{A}_-)^k}{{\rho^{(p,q)}(k)}}.
\end{align}
Using the DOOT method \cite{doot} and the completeness of Fock states, we derive the projector
 $\mid 0><0 \mid$ associated with the ground state $\mid 0>$ as follows
 \begin{align}\label{eq:cz9k}
     &\sum_{k=0}^{\infty}\mid k><k\mid=1\nonumber\\
    &\Rightarrow\sum_{k=0}^{\infty}\frac{1}{\rho^{(p,q)}(k)}\#(\mathcal{A}_+)^k\mid 0><0 \mid (\mathcal{A}_-)^k\#=1\nonumber\\
    &\Rightarrow \mid 0><0 \mid \sum_{k=0}^{\infty}\frac{1}{\rho^{(p,q)}(k)}\#  (\mathcal{A}_+\mathcal{A}_-)^{k} \#=1\nonumber\\
    &\Rightarrow \mid 0><0 \mid \#\mathcal{N}^{(p,q)}(\mathcal{A}_+\mathcal{A}_-)\#=1\nonumber\\
    &\Rightarrow \mid 0><0 \mid=\frac{\prod\limits_{l=1}^p\Gamma(a_l)}{\prod\limits_{r=1}^q \Gamma(b_r)\# {}_{p}\psi_{q}\left[\begin{matrix}&(a_1,A_1),\ldots &(a_p,A_p);\\& (b_1,B_1),\ldots&(b_q,B_q);&\end{matrix}\mathcal{A}_+\mathcal{A}_-\right]\#}.
      \end{align}
      Next, we determine the weight function $W(|z|^2)$ ensuring that the Fox-Wright coherent states satisfy the resolution of unity
\begin{align}\label{eq:cs9m}
    &\int d \mu(z)\mid z><z\mid=1.
\end{align}
Using \eqref{eq:z9x},\eqref{eq:cz9k} and \eqref{eq:cs9m}, we have
\begin{align}\label{eq:csl8}
    &\int  \frac{d\theta}{2\pi}d(|z|^2)W(|z|^2) \frac{1}{\mathcal{N}^{(p,q)}(\mid z\mid^2)}\sum_{k=0}^\infty \frac{\#(z\mathcal{A}_+)^k(z^\ast\mathcal{A}_-)^k\#}{{[\rho^{(p,q)}(k)]^2}}\mid 0><0\mid=1\nonumber\\
    &\Rightarrow \left[\int_{0}^{2\pi} \frac{d\theta}{2\pi}\right]\;\int_{0}^{\infty}d(|z|^2)\frac{W(|z|^2)}{\mathcal{N}^{(p,q)}\left(\mid z\mid^2\right)}\sum_{k=0}^\infty \frac{(|z|^2)^k}{{[\rho^{(p,q)}(k)]^2}}\#(\mathcal{A}_+\mathcal{A}_-)^n\#=\#\mathcal{N}^{(p,q)}(\mathcal{A}_+\mathcal{A}_-)\#\nonumber\\
    &\Rightarrow \sum_{k=0}^\infty \frac{\#(\mathcal{A}_+\mathcal{A}_-)^k\#}{{[\rho^{(p,q)}(k)]^2}}\int_{0}^{\infty}d(|z|^2)\frac{W(|z|^2)}{\mathcal{N}^{(p,q)}\left(\mid z\mid^2\right)}(|z|^2)^k=\#\mathcal{N}^{(p,q)}(\mathcal{A}_+\mathcal{A}_-)\#.
    \end{align}
    Setting the value of the normalization function in \eqref{eq:csl8} and comparing both side we obtain
    \begin{align}\label{eq:s9d}
        \int_{0}^{\infty}d(|z|^2)\frac{W(|z|^2)}{\mathcal{N}^{(p,q)}\left(\mid z\mid^2\right)}(|z|^2)^k=\rho^{(p,q)}(k) = \Gamma(k+1)\frac{\prod\limits_{l=1}^p\Gamma(a_l)}{\prod\limits_{r=1}^q \Gamma(b_r)}\frac{\prod\limits_{r=1}^{q}\Gamma(b_r+kB_r)}{\prod\limits_{l=1}^{p}\Gamma(a_l+kA_l)}.
    \end{align}
 Substituting $k=s-1$ in \eqref{eq:s9d} and using the Mellin transform of a $H$-function \cite{h_transform}:
  \begin{align}\label{eq:mlt7}
      \int\limits_{0}^\infty 
H_{m, n}^{l_1, l_2} \left[ x \;\middle|\; \begin{matrix} (a_1,\alpha_1), \ldots, (a_m,\alpha_m) \\ (b_1,\beta_1), \ldots, (b_n,\beta_n) \end{matrix} \right]x^{s-1} dx=\frac{\prod\limits_{j=1}^{l_1}\Gamma(b_j+s\beta_j)\prod\limits_{i=1}^{l_2}\Gamma(1-a_i-s\alpha_i)}{\prod\limits_{j=l_1+1}^n\Gamma(1-b_j-s\beta_j)\prod\limits_{i=l_2+1}^{m}\Gamma(a_i+s\alpha_i)},
  \end{align}
we obtain 
\begin{align*}
   \frac{W(|z|^2)}{\mathcal{N}^{(p,q)}\left(\mid z\mid^2\right)}&=\frac{\prod\limits_{l=1}^p\Gamma(a_l)}{\prod\limits_{r=1}^q \Gamma(b_r)}H_{p,q+1}^{q+1,0} \left[ |z|^2 \;\middle|\; \begin{matrix} &(a_{1}-A_{1},A_{1}),&(a_{2}-A_{2},A_{2}) \ldots, &(a_{p}-A_{p},A_{p}) \\ &(0,1),&(b_{1}-B_{1},B_{1}), \ldots, &(b_{q}-B_{q},B_{q}) \end{matrix} \right].
\end{align*}
 Therefore, weight function is given by
\begin{align*}
    W(|z|^2)&={}_{p}\psi_{q}\left[\begin{matrix}&(a_1,A_1),\ldots &(a_p,A_p);\\& (b_1,B_1),\ldots&(b_q,B_q);&\end{matrix}|z|^2\right]\\&\hspace{10pt}H_{p,q+1}^{q+1,0} \left[ |z|^2 \;\middle|\; \begin{matrix} &(a_{1}-A_{1},A_{1}),&(a_{2}-A_{2},A_{2}) \ldots, &(a_{p}-A_{p},A_{p}) \\ &(0,1),&(b_{1}-B_{1},B_{1}), \ldots, &(b_{q}-B_{q},B_{q}) \end{matrix} \right]>0,
\end{align*}
where the positivity of the corresponding $H$-function has been established in \cite{positive_h}.
Hence, the corresponding integration measure is given by
\begin{align*}
    d\mu(z)&=\frac{d\theta}{2\pi}d(|z|^2){}_{p}\psi_{q}\left[\begin{matrix}&(a_1,A_1),\ldots &(a_p,A_p);\\& (b_1,B_1),\ldots&(b_q,B_q);&\end{matrix}|z|^2\right]\\&\hspace{10pt}H_{p,q+1}^{q+1,0} \left[ |z|^2 \;\middle|\; \begin{matrix} &(a_{1}-A_{1},A_{1}),&(a_{2}-A_{2},A_{2}) \ldots, &(a_{p}-A_{p},A_{p}) \\ &(0,1),&(b_{1}-B_{1},B_{1}), \ldots, &(b_{q}-B_{q},B_{q}) \end{matrix} \right].
\end{align*}
D, Popov et al. \cite{coh_cont}, constructed coherent states for the continuous spectrum of a quantum system using a discrete-to-continuous limit. The same technique is employed here to define Fox-Wright coherent states for the continuous spectrum. Apply discrete continuous limit $d\rightarrow c$ in \eqref{eq:a8f} and \eqref{eq:vt9}, we obtain
\begin{align}\label{eq:a6i}
    \mathcal{\Tilde{N}}^{(p,q)}(\zeta)=\lim_{d\rightarrow c}\mathcal{N}^{(p,q)}(\zeta)&=\frac{\prod\limits_{r=1}^q \Gamma(b_r)}{\prod\limits_{l=1}^p\Gamma(a_l)}\int\limits_{0}^{\infty}dE\frac{\prod\limits_{l=1}^{p}\Gamma(a_l+EA_l)}{\prod\limits_{r=1}^{q}\Gamma(b_r+EB_r)}\frac{\zeta^E}{\Gamma(E+1)}
\end{align}
and 
\begin{align}
    \Tilde{\rho}^{(p,q)}(E)=\lim_{d\rightarrow c}\rho^{(p,q)}(k)=\Gamma(E+1)\frac{\prod\limits_{l=1}^p\Gamma(a_l)}{\prod\limits_{r=1}^q \Gamma(b_r)}\frac{\prod\limits_{r=1}^{q}\Gamma(b_r+EB_r)}{\prod\limits_{l=1}^{p}\Gamma(a_l+EA_l)}.
\end{align}
Setting $A_l=B_r=1$ in \eqref{eq:a6i} for $l=1,\ldots,p;r=1,\ldots,q$, we get
\begin{align}\label{eq:s9e}
    \mathcal{\Tilde{N}}^{(p,q)}(\zeta)=\frac{\prod\limits_{r=1}^q \Gamma(b_r)}{\prod\limits_{l=1}^p\Gamma(a_l)}\int\limits_{0}^{\infty}dE\frac{\prod\limits_{l=1}^{p}\Gamma(a_l+E)}{\prod\limits_{r=1}^{q}\Gamma(b_r+E)}\frac{\zeta^E}{\Gamma(E+1)}=\int\limits_{0}^{\infty}dE\frac{\prod\limits_{l=1}^{p}(a_l)_E}{\prod\limits_{r=1}^{q}(b_r)_E}\frac{\zeta^E}{\Gamma(E+1)}
\end{align}
where $(a)_E$ is the Pochhammer symbols given by $(a)_E=\frac{\Gamma(a+E)}{\Gamma(a)}$. It is observe that the integral in \eqref{eq:s9e} is generalized multi-parameter $\nu$-function, which has already been defined in \cite{int_mittage}. Since the generalized multi-parameter $\nu$-function represents a specific case of the integral \eqref{eq:a6i}, we refer to it as the FW-generalized multi-parameter $\nu$-function and denote it by $\mathcal{V}^{(p,q)}(\zeta)$. Therefore, $\mathcal{\Tilde{N}}^{(p,q)}(\zeta)=\mathcal{V}^{(p,q)}(\zeta)$. Similarly, apply discrete continuous limit $d\rightarrow c$ in \eqref{eq:gh9}, we have
\begin{align}
    \mid \Tilde{z}>=\lim_{d\rightarrow c}\mid z> =\frac{1}{\sqrt{\mathcal{V}^{(p,q)}(|z|^2)}}\int\limits_{0}^\infty dE \frac{z^E}{\sqrt{\Tilde{\rho}^{(p,q)}(E)}}\mid E>,
\end{align}
which is the expression of Fox-Wright coherent states for continuous spectrum. The overlap of two coherent states for continuous spectrum is given by
\begin{equation}
< \Tilde{z} \mid \Tilde{z}'>=\lim_{d\rightarrow c}< z\mid z'>=\frac{\mathcal{V}^{(p,q)}(z*z')}{\sqrt{\mathcal{V}^{(p,q)}(\mid z\mid^2)}\sqrt{\mathcal{V}^{(p,q)}(\mid z'\mid^2)}},\nonumber
\end{equation}
which implies the Fox-Wright coherent states for continuous spectrum are normalized.  Now, apply discrete continuous limit $d\rightarrow c$ in \eqref{eq:csb7}, we have
\begin{align}
    \mid E>=\frac{1}{\sqrt{\Tilde{\rho}^{(p,q)}(E)}}(\mathcal{A}_+)^E\mid 0>\;\; \mbox{and} \;\;< E \mid=\frac{1}{\sqrt{\Tilde{\rho}^{(p,q)}(E)}} <0 \mid (\mathcal{A}_-)^E,
\end{align}
and similarly to the discrete case, by employing the completeness relation for the continuous spectrum, We can directly derive the projector corresponding to the vacuum state $\mid 0><0 \mid$ in the following form 
\begin{align*}
    \int\limits_{0}^\infty dE\mid E><E \mid=1\Rightarrow\mid 0><0 \mid=\frac{1}{\# \mathcal{V}^{(p,q)}(\mathcal{A}_+\mathcal{A}_-)\#}.
\end{align*}
Similarly, we can easily prove the Fox-Wright coherent states for continuous spectrum satisfies resolution of unity
\begin{align*}
    \int d \Tilde{\mu}(z)\mid \Tilde{z}><\Tilde{z}\mid=1,
\end{align*}
where the integration measure for continuous spectrum is given by 
\begin{align*}
    d\Tilde{\mu}(z)&=\frac{\prod\limits_{l=1}^p\Gamma(a_l)}{\prod\limits_{r=1}^q \Gamma(b_r)}\frac{d\theta}{2\pi}d(|z|^2)\mathcal{V}^{(p,q)}(|z|^2)\\&\hspace{10pt}H_{p,q+1}^{q+1,0} \left[ |z|^2 \;\middle|\; \begin{matrix} &(a_{1}-A_{1},A_{1}),&(a_{2}-A_{2},A_{2}) \ldots, &(a_{p}-A_{p},A_{p}) \\ &(0,1),&(b_{1}-B_{1},B_{1}), \ldots, &(b_{q}-B_{q},B_{q}) \end{matrix} \right].
\end{align*}
For particular choices of parameters, the following distinct cases can be identified:
\begin{remark}\label{rm1}
    By setting $A_l=B_r=1$ for $l=1,2,\ldots,p;r=1,2,\ldots,q$, we get generalized hypergeometric coherent states $\mid z >$ with the corresponding normalization function
    \begin{align*}
        \mathcal{N}^{(p,q)}(\zeta)=\frac{\prod\limits_{r=1}^q \Gamma(b_r)}{\prod\limits_{l=1}^p\Gamma(a_l)}\sum_{k=0}^{\infty}\frac{\prod\limits_{l=1}^{p}\Gamma(a_l+k)}{\prod\limits_{r=1}^{q}\Gamma(b_r+k)}\frac{\zeta^k}{k!}
        =\sum_{k=0}^{\infty}\frac{\prod\limits_{l=1}^{p}(a_l)_k}{\prod\limits_{r=1}^{q}(b_r)_k}\frac{\zeta^k}{k!}={}_{p}F_q\left[\begin{matrix}&a_1,\ldots &a_p;\\&b_1,\ldots&b_q;&\end{matrix}\zeta\right]
    \end{align*}
    and weight function is given by
    \begin{align*}
        W(|z|^2)&={}_{p}\psi_{q}\left[\begin{matrix}&(a_1,1),\ldots &(a_p,1);\\& (b_1,1),\ldots&(b_q,1);&\end{matrix}|z|^2\right]H_{p,q+1}^{q+1,0} \left[ |z|^2 \;\middle|\; \begin{matrix} &(a_{1}-1,1),&(a_{2}-1,1) \ldots, &(a_{p}-1,1) \\ &(0,1),&(b_{1}-1,1), \ldots, &(b_{q}-1,1) \end{matrix}\right]\\
        &=\frac{\prod\limits_{l=1}^p\Gamma(a_l)}{\prod\limits_{r=1}^q \Gamma(b_r)}{}_{p}F_q\left[\begin{matrix}&a_1,\ldots &a_p;\\&b_1,\ldots&b_q;&\end{matrix}|z|^2\right]G_{p,q+1}^{q+1,0} \left[ |z|^2 \;\middle|\; \begin{matrix} &a_{1}-1,&a_{2}-1 \ldots, &a_{p}-1 \\ &0,&b_{1}-1, \ldots, &b_{q}-1 \end{matrix}\right]
    \end{align*}
    T. Appl et al. \cite{gh_cs}, constructed coherent states associated with the generalized hypergeometric function, which coincide with our results. This confirms the validity of our results.
\end{remark}
\begin{remark}
    By choosing $p=q=a_1=A_1=1$ the Mittage-Leffler coherent states $\mid z>$ are obtained, together with their associated normalization function
    \begin{align*}
        \mathcal{N}^{(p,q)}(\zeta)=\frac{\prod\limits_{r=1}^q \Gamma(b_r)}{\prod\limits_{l=1}^p\Gamma(a_l)} {}_{p}\psi_{q}\left[\begin{matrix}&(1,1);\\& (b_1,B_1);&\end{matrix}\zeta\right]=\frac{\prod\limits_{r=1}^q \Gamma(b_r)}{\prod\limits_{l=1}^p\Gamma(a_l)} E_{B_1,b_1}(\zeta)
    \end{align*}
    and weight function is given by
    \begin{align*}
         W(|z|^2)&=E_{B_1,b_1}(|z|^2)H_{1,2}^{2,0} \left[ |z|^2 \;\middle|\; \begin{matrix} & &(0,1) \\ &(0,1),&(b_{1}-B_1,B_1) \end{matrix}\right].\\
    \end{align*}
\end{remark}
\section{Bicomplex Fox-Wright function }\label{sec3}
\subsection{Preliminaries on Bicomplex numbers}
   The bicomplex number system is an extension of the complex number system, formed by introducing additional imaginary units together with complex coefficients. C. Segre \cite{le} introduced the concept of bicomplex numbers in 1892. The theory of bicomplex numbers has attracted growing attention in recent years. The set of bicomplex numbers emerge as a commutative alternative to the noncommutative quaternion skew field, which is another four-dimensional real space. Unlike quaternions, bicomplex numbers have commutative multiplication and form a ring with
zero divisors. Extensive studies have explored the properties of bicomplex numbers, covering their algebraic and geometric aspects along with several applications see \cite{bc_number,multicomplex,hyperbolic}.
The set $\mathbb{BC}$ of bicomplex numbers is given by
\begin{align*}\mathbb{BC}&=\{Z=a+jb=x_1+ix_2+jx_3+ijx_4:x_1,x_2,x_3,x_4\in \mathbb{R}\},
\end{align*}
where $i$ and $j$ satisfies the conditions $i^2=j^2=-1,ij=ji$ with $a=x_1+ix_2,b=x_3+ix_4$.
Every bicomplex number $Z=a+jb\in \mathbb{BC} $ admits a unique representation of the following form  $$ Z=(a-ib)e_1+(a+ib)e_2=z_1e_1+z_2 e_2,$$ where $e_1=\frac{1+ij}{2}$  and \; $e_2=\frac{1-ij}{2}$, which satisfies the identities $e_1+e_2=1, e_1-e_2=ij, e_1^{2}=e_1$ and $e_2^{2}=e_2$. This form is known as the idempotent representation of the bicomplex number $Z\in\mathbb{BC}$. 
The space $\mathbb{BC}$ forms a commutative ring with unity that contains zero divisors. The collection of all zero divisors is given by \cite{bc_holomorphic} : $$\mathbb{O}_2=\{Z=a+jb:a^2+b^2=0\},$$ 
and its elements are known as singular elements.
 
The hyperbolic number system is defined as
$\mathbb{D}=\{Z=x_1+ijx_4:x_1,x_4\in\mathbb{R}\}$. Inside $\mathbb{D}$ the subsets of non-negative and non-positive hyperbolic numbers can be introduced as follows:
\begin{align*}
       \mathbb{D}^+=\{Z=x_1+ijx_4:x_1^2-x_4^2\geq0,x_1\geq0\}\quad \mbox{and}\quad \mathbb{D}^-=\{Z=x_1+ijx_4:x_1^2-x_4^2\geq0,x_1\leq0\},
   \end{align*}respectively. Furthermore, a partial order relation $<_h$ is defined on the set $\mathbb{D}$ and if $Z,W\in\mathbb{D}$ and $Z<_hW$, it implies that $W-Z\in\mathbb{D}^+$ and $Z-W\in\mathbb{D}^-$. 
 The hyperbolic norm $|\cdot|_h$ for any bicomplex number $Z=z_1e_1+z_2 e_2$, is given by\cite{bc_holomorphic}: $$|Z|_h=|z_1|e_1+|z_2|e_2.$$
 Using the hyperbolic norm, the hyperbolic ball with center $p$ and radius $R$ is defined as follows $B_h(p,R)=\{Z:|Z-p|<_hR\}$.
 The following expression represents the Euler product form of the bicomplex gamma function \cite{gamma}:
    \begin{align*}
        \Gamma_b(W)=\frac{1}{We^{\gamma W}}\prod_{n=1}^{\infty}\left(\left(\frac{n}{n+W}\right)\exp\left(\frac{W}{n}\right)\right),
    \end{align*}
    where $W=c+jd=w_1e_1+w_2e_2\in{\mathbb{BC}}$,
   with $c\not=-\frac{(p+q)}{2}$, $d\not=i\frac{(q-p)}{2}$ where  $p,q\in \mathbb{N}\cup \{0\}$ and $\gamma$ denotes the Euler constant \cite{sp_function}. Additionally, idempotent representation of the bicomplex gamma function is given by 
    \begin{align}\label{eq:g}
        \Gamma_b(W)=\Gamma(w_1) e_1+ \Gamma(w_2) e_2.
    \end{align} 
 
    \begin{lemma}\rm{\cite{higher_trans}}{(Stirling’s formula)}\label{asygamma}
        The asymptotic formula for the Gamma function is given as follows
        \begin{align*}
    \Gamma(w)=\sqrt{2\pi}\;w^{w-\frac{1}{2}}e^{-w}\left[1+O\left(\frac{1}{w}\right)\right],\;\;\text{as}\;\;|w|\rightarrow\infty. 
\end{align*}
    \end{lemma}
\subsection{Bicomplex Fox-Wright function}
In the following, we introduce the bicomplex Fox-Wright function and examine its existence within the bicomplex domain.  We introduce the bicomplex Fox–Wright function in the following form
 \begin{align}\label{eq:fw}
        {}_{m}\Psi_{n}\left[\begin{matrix}&(\mu_1,M_1),\ldots &(\mu_m,M_m);\\&(\nu_1,N_1),\ldots&(\nu_n,N_n);&\end{matrix}Z\right]=
  \sum_{k=0}^{\infty}\frac{\prod\limits_{i=1}^{m}\Gamma_b(\mu_i+kM_i)}{\prod\limits_{j=1}^{n}\Gamma_b(\nu_j+kN_j)}\frac{Z^k}{k!},
    \end{align}
where $Z=z_1e_1+z_2e_2,\mu_i=\mu_{1i}e_1+\mu_{2i}e_2,\nu_j=\nu_{1j}e_1+\nu_{2j}e_2\in\mathbb{BC}$ and $M_i=M_{1i}e_1+M_{2i}e_2,N_j=N_{1j}e_1+N_{2j}e_2\in\mathbb{D}^+-\mathbb{O}_2$, for $i=1,2,\ldots,m;\;j=1,2,\ldots,n$.

 \begin{theo}\label{th:a5}
     Assume that $Z=z_1e_1+z_2e_2,\mu_i=\mu_{1i}e_1+\mu_{2i}e_2,\nu_j=\nu_{1j}e_1+\nu_{2j}e_2\in\mathbb{BC}$, $M_i=M_{1i}e_1+M_{2i}e_2,N_j=N_{1j}e_1+N_{2j}e_2\in\mathbb{D}^+-\mathbb{O}_2$ for $i=1,2,\ldots,m$; $j=1,2,\ldots,n$, and
     \begin{align*}
         &\Upsilon=\Upsilon_1e_1+\Upsilon_2e_2=\sum_{j=1}^nN_j-\sum_{i=1}^mM_i,\\
         &\mathcal{V}=\mathcal{V}_1e_1+\mathcal{V}_2e_2=\prod_{j=1}^n|N_j|_h^{N_j}\prod_{i=1}^m|M_i|_h^{-M_i}\\
         &\Lambda=\Lambda_1+j\Lambda_2=\lambda_1e_1+\lambda_2e_2=\sum_{j=1}^{n}\nu_j-\sum_{i=1}^m\mu_i-\frac{n-m}{2}.
     \end{align*}
    \begin{enumerate}
        \item [(i)] If $\Upsilon>_h-1$, then the series \eqref{eq:fw} is absolutely hyperbolic convergent in the bicomplex space $\mathbb{BC}$.
        \item [(ii)] If $\Upsilon_1=-1$ and $\Upsilon_2>-1$, then the series \eqref{eq:fw} absolutely hyperbolic converges inside $\mathcal{B}_1(0,R)=\{z_1e_1+z_2e_2:|z_1|<\mathcal{V}_1\; \text{and}\;z_2\in\mathbb{C}\; \}$.
        \item [(iii)] If $\Upsilon_1>-1$ and $\Upsilon_2=-1$, then the series \eqref{eq:fw} converges absolutely in the hyperbolic sense within $\mathcal{B}_2(0,R)=\{z_1e_1+z_2e_2:z_1\in\mathbb{C}\; \text{and}\;|z_2|<\mathcal{V}_2\}$.
        \item [(iv)] If $\Upsilon_1=-1$ and $\Upsilon_2<-1$, then the series \eqref{eq:fw} absolutely hyperbolic convergent within the set $\mathcal{B}_3(0,R)=\{z_1e_1+z_2e_2:|z_1|<\mathcal{V}_1\; \text{and}\;z_2=0\; \}$.
        \item [(v)] If $\Upsilon_1<-1$ and $\Upsilon_2=-1$, then the series \eqref{eq:fw} absolutely hyperbolic converges inside $\mathcal{B}_4(0,R)=\{z_1e_1+z_2e_2:z_1=0\; \text{and}\;|z_2|<\mathcal{V}_2\}$.
        \item [(vi)] If $\Upsilon_1>-1$ and $\Upsilon_2<-1$, then the series \eqref{eq:fw} absolutely hyperbolic convergent within the set $\mathcal{B}_5(0,R)=\{z_1e_1+z_2e_2:z_1\in\mathbb{C}\; \text{and}\;z_2=0\}$.
        \item [(vii)] If $\Upsilon_1<-1$ and $\Upsilon_2>-1$, then the series \eqref{eq:fw} is absolutely hyperbolically convergent in the set $\mathcal{B}_6(0,R)=\{z_1e_1+z_2e_2:z_1=0\; \text{and}\;z_2\in\mathbb{C}\}$.
        \item [(viii)] If $\Upsilon=-1$, then the series \eqref{eq:fw} absolutely hyperbolic converges inside  $\mathcal{B}_h(0,\mathcal{V}
        )=\{z_1e_1+z_2e_2:|z_1|<\mathcal{V}_1\; \text{and}\;|z_2|<\mathcal{V}_2\}$. Moreover, on the boundary of the hyperbolic ball $\mathcal{B}_h(0,\mathcal{V})$ the series \eqref{eq:fw} hyperbolic uniformly and absolutely convergent if
        \begin{align}\label{eq:so}
        \Re(\Lambda_1)-\frac{1}{2}>|\Im(\Lambda_2)|.
        \end{align}
        \item [(ix)] If $\Upsilon<_h-1$, then the series \eqref{eq:fw}is divergent for all nonzero elements $Z$ in $\mathbb{BC}$.
    \end{enumerate}
\end{theo}

\begin{proof}
Let $R=R_1e_1+R_2e_2$ be defined as the hyperbolic radius of convergent of the bicomplex power series \eqref{eq:fw}
\begin{align*}
\sum_{k=0}^{\infty}\frac{\prod\limits_{i=1}^{m}\Gamma_b(\mu_i+kM_i)}{\prod\limits_{j=1}^{n}\Gamma_b(\nu_j+kN_j)}\frac{Z^k}{k!}=\sum_{k=0}^{\infty}a_kZ^k=\sum_{k=0}^{\infty}a_{1k}z_1^ke_1+\sum_{k=0}^{\infty}a_{2k}z_2^ke_2.
\end{align*}
By applying Lemma \ref{asygamma}, we obtain the asymptotic behavior of $\Gamma(\mu_{pi}+kM_{pi})$ and $\Gamma(\nu_{pj}+kN_{pj})$ for $p=1,2$, as given by
\begin{align}\label{eq:cp}
    \Gamma(\mu_{pi}+kM_{pi})\sim \sqrt{2\pi} \left(\mu_{pi}+kM_{pi}\right)^{\mu_{pi}+kM_{pi}-\frac{1}{2}} e^{-(\mu_{pi}+kM_{pi})}\;\; \text{as}\;\;k\rightarrow\infty,
\end{align}
 and 
\begin{align}\label{eq:sp}
    \Gamma(\nu_{pj}+kN_{pj})\sim\sqrt{2\pi} \left(\nu_{pj}+kN_{pj}\right)^{\nu_{pj}+kN_{pj}-\frac{1}{2}} e^{-(\nu_{pj}+kN_{pj})} \;\;\text{as}\;\;k\rightarrow\infty,
\end{align}
where $i=1,2,\ldots,m;j=1,2,\ldots,n$.
Then by using hyperbolic root test \cite{bc_root}, \eqref{eq:cp}and \eqref{eq:sp}, we obtain
\begin{align}\label{eq:mi}
     \frac{1}{R}&=\lim_{k\to \infty} \sup|a_k|^{\frac{1}{k}}_h\nonumber\\
    &=\lim_{k\to \infty} \sup|a_{1k}|^{\frac{1}{k}}e_1+\lim_{k\to \infty} \sup|a_{2k}|^{\frac{1}{k}}e_2\nonumber\\
    &=\sum_{p=1}^2\lim_{k\to\infty}\left|\frac{\prod\limits_{i=1}^{m}\Gamma(\mu_{pi}+kM_{pi})}{\prod\limits_{j=1}^{n}\Gamma(\nu_{pj}+kN_{pj})\Gamma(k+1)}\right|^{\frac{1}{k}}e_p\nonumber\\
    &=\sum_{p=1}^2\lim_{k\to\infty}\left|\frac{\prod\limits_{i=1}^{m}\sqrt{2\pi} \left(\mu_{pi}+kM_{pi}\right)^{\mu_{pi}+kM_{pi}-\frac{1}{2}} e^{-(\mu_{pi}+kM_{pi})}}{\prod\limits_{j=1}^{n}\sqrt{2\pi} \left(\nu_{pj}+kN_{pj}\right)^{\nu_{pj}+kN_{pj}-\frac{1}{2}} e^{-(\nu_{pj}+kN_{pj})}\times\sqrt{2\pi}\;(k+1)^{k+\frac{1}{2}}e^{-k-1}}\right|^{\frac{1}{k}}e_p\nonumber\\
    &=\sum_{p=1}^2\lim_{k\to\infty}\frac{\prod\limits_{i=1}^{m} \left(k|M_{pi}|\right)^{M_{pi}} e^{-M_{pi}}}{\prod\limits_{j=1}^{n} \left(k|N_{pj}|\right)^{N_{pj}} e^{-N_{pj}}\times ke^{-1}}e_p\nonumber\\
    &=\sum_{p=1}^2\frac{e^{\Upsilon_p+1}}{\mathcal{V}_p}\lim_{k\to\infty}\frac{1}{k^{\Upsilon_p+1}}e_p.
\end{align}
From \eqref{eq:mi}, we get 
\begin{enumerate}
    \item [(i)] If $\Upsilon_p>-1$ for $p=1,2$, we obtain $R_1=\infty$ and $R_2=\infty$ which leads to $R=\infty$. Therefore the series \eqref{eq:fw} is absolutely hyperbolic convergent in $\mathbb{BC}$.
    \item [(ii)] If $\Upsilon_1=-1$ and $\Upsilon_2>-1$, then $R_1=\mathcal{V}_1$ and $R_2=\infty$, so the series \eqref{eq:fw} absolutely hyperbolic converges inside $\mathcal{B}_1(0,R)=\{Z=z_1e_1+z_2e_2:|z_1|<\mathcal{V}_1\; \text{and}\;z_2\in\mathbb{C}\; \}$.
    \item [(iii)] If $\Upsilon_1>-1$ and $\Upsilon_2=-1$, then $R_1=\infty$ and $R_2=\mathcal{V}_2$ so, the series \eqref{eq:fw} absolutely hyperbolic converges inside $\mathcal{B}_2(0,R)=\{Z=z_1e_1+z_2e_2:Z_1\in\mathbb{C}\; \text{and}\;|z_2|<\mathcal{V}_2\}$.
    \item [(iv)] If $\Upsilon_1=-1$ and $\Upsilon_2<-1$, then $R_1=\mathcal{V}_1$ and $R_2=0$. Therefore the series \eqref{eq:fw} absolutely hyperbolic converges in $\mathcal{B}_3(0,R)=\{Z=z_1e_1+z_2e_2:|z_1|<\mathcal{V}_1\; \text{and}\;z_2=0\; \}$. This set corresponds to a disk $D_{e_1}\left(0,\frac{\mathcal{V}_1}{\sqrt{2}}\right)$ in the two-dimensional plane $\mathbb{BC}_{e_1}$ with center at the origin and radius $\frac{\mathcal{V}_1}{\sqrt{2}}$ see Fig. \ref{fig:fwdisk1}(A).
    \item [(v)] If $\Upsilon_1<-1$ and $\Upsilon_2=-1$, then $R_1=0$ and $R_2=\mathcal{V}_2$; consequently, the series \eqref{eq:fw} absolutely hyperbolic converges in $\mathcal{B}_4(0,R)=\{z_1e_1+z_2e_2:z_1=0\; \text{and}\;|z_2|<\mathcal{V}_2\}$, this set represent a disk $D_{e_2}\left(0,\frac{\mathcal{V}_2}{\sqrt{2}}\right)$ in another two-dimensional plane $\mathbb{BC}_{e_2}$ with center at the origin and radius $\frac{\mathcal{V}_2}{\sqrt{2}}$ see Fig. \ref{fig:fwdisk1}(B).
    \begin{figure}[ht!]
   \centering
   \begin{subfigure}{0.35\linewidth}
     \includegraphics[width=\linewidth]{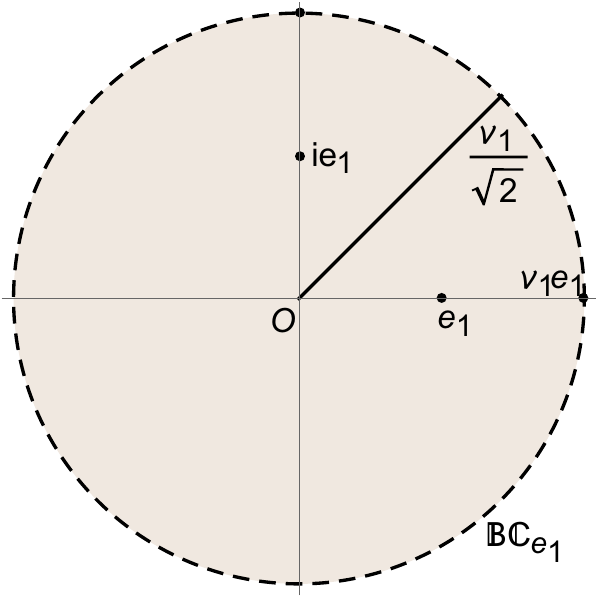}
    \caption{Disk $D_{e_1}\left(0,\frac{\mathcal{V}_1}{\sqrt{2}}\right)$}
  \end{subfigure}
   \begin{subfigure}{0.35\linewidth}
     \includegraphics[width=\linewidth]{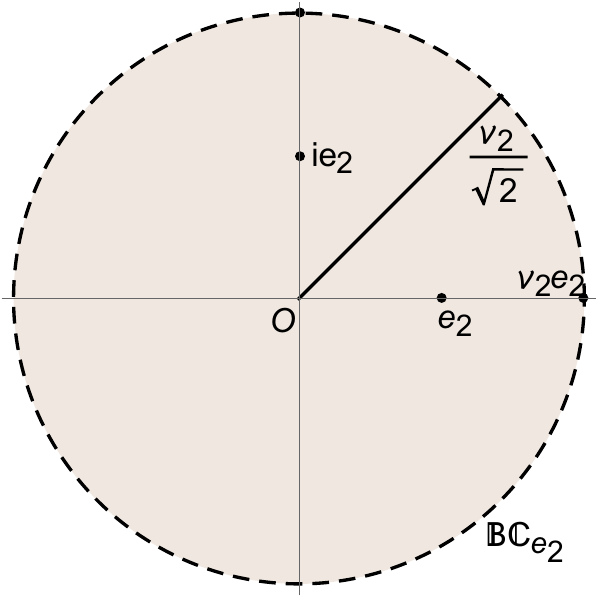}
     \caption{Disk $D_{e_2}\left(0,\frac{\mathcal{V}_2}{\sqrt{2}}\right)$}
   \end{subfigure}
   \caption{Image of $D_{e_1}(0,\frac{\mathcal{V}_1}{\sqrt{2}})$ and $D_{e_2}(0,\frac{\mathcal{V}_2}{\sqrt{2}})$ in $\mathbb{BC}_{e_1}$ and $\mathbb{BC}_{e_2}$ respectively}
   \label{fig:fwdisk1}
\end{figure}
    \item [(vi)]  If $\Upsilon_1>-1$ and $\Upsilon_2<-1$, then $R_1=\infty$ and $R_2=0$; so, the series \eqref{eq:fw} absolutely hyperbolic converges inside $\mathcal{B}_5(0,R)=\{z_1e_1+z_2e_2:z_1\in\mathbb{C}\; \text{and}\;z_2=0\}$, is a usual two dimensional plane that passes through origin and parallel to $\mathbb{BC}_{e_1}$.
    \item [(vii)] If $\Upsilon_1<-1$ and $\Upsilon_2>-1$, then $R_1=0$ and $R_2=\infty$; hence, the series \eqref{eq:fw} absolutely hyperbolic converges in the set $\mathcal{B}_6(0,R)=\{z_1e_1+z_2e_2:z_1=0\; \text{and}\;z_2\in\mathbb{C}\}$, is a usual two dimensional plane that passes through origin and parallel to $\mathbb{BC}_{e_2}$.
    \item [(viii)] If $\Upsilon=-1$, then $R_1=\mathcal{V}_1$ and $R_2=\mathcal{V}_2$; so, the series \eqref{eq:fw} hyperbolic converges absolutely inside the hyperbolic ball $\mathcal{B}_h(0,R)=\{Z=z_1e_1+z_2e_2:|z_1|<\mathcal{V}_1\; \text{and}\;|z_2|<\mathcal{V}_2\; \}$ and diverges on the complement of its closure. The hyperbolic ball $\mathcal{B}_h(0,R)$ is represented as the Cartesian product of the two disk $\mathrm{D}_{e_1}\left(0,\frac{\mathcal{V}_1}{\sqrt{2}}\right)$ and $\mathrm{D}_{e_2}\left(0,\frac{\mathcal{V}_2}{\sqrt{2}}\right)$, which are situated in the planes $\mathbb{BC}_{e_1}$ and $\mathbb{BC}_{e_2}$ respectively see Fig. \ref{fig:fwdisk2}. Although Fig. \ref{fig:fwdisk2} depicts a four-dimensional object, it is illustrated as the interior of a three-dimensional object for the sake of abstract visualization. 
     \begin{figure}[ht!]
   \centering
   \begin{subfigure}{0.45\linewidth}
     \includegraphics[width=\linewidth]{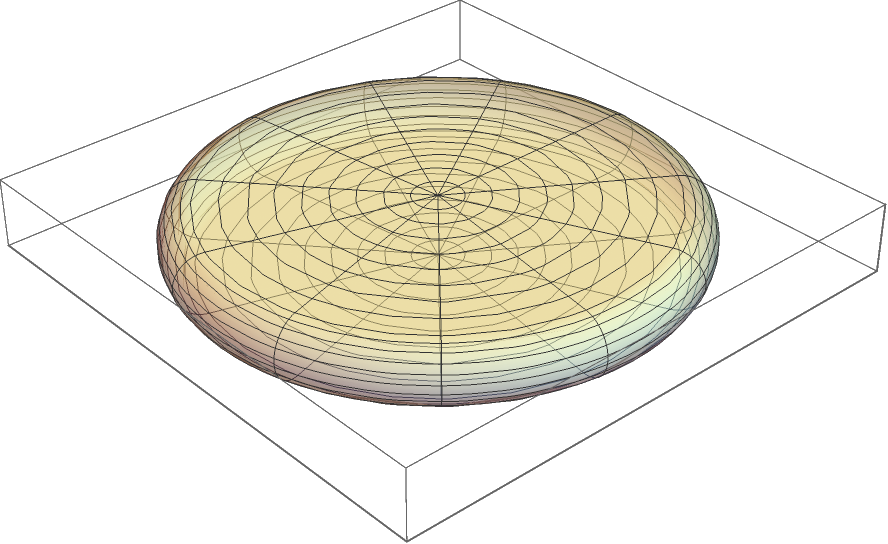}
   \end{subfigure}
   \caption{Outer surface of the ball $\mathcal{B}_h(0,R)$}
   \label{fig:fwdisk2}
\end{figure}
    \item [(ix)] If $\Upsilon<_h-1$, then $R=0$; therefore, the series \eqref{eq:fw} divergent for every nonzero $Z\in \mathbb{BC}$.
\end{enumerate}
Next, we investigate the behavior of the series \eqref{eq:fw} along the boundary of the hyperbolic ball $\mathcal{B}_h(0,R)=\{z_1e_1+z_2e_2:|z_1|<\mathcal{V}_1\; \text{and}\;|z_2|<\mathcal{V}_2\; \}$ corresponding to the case $\Upsilon=-1$. In this case $|z_1|=\mathcal{V}_1$, $|z_2|=\mathcal{V}_2$ and from the inequality \eqref{eq:so} we obtain
\begin{align}\label{eq:x4}
    \Re(\Lambda_1)+\Im(\Lambda_2)>\frac{1}{2}\implies \Re(\lambda_1)>\frac{1}{2}
\end{align}
and
\begin{align}\label{eq:s6}
    \Re(\Lambda_1)-\Im(\Lambda_2)>\frac{1}{2}\implies \Re(\lambda_2)>\frac{1}{2}.
\end{align}
Using \eqref{eq:cp} and \eqref{eq:sp}, we get
\begin{align}
    |a_kZ^k|_h&=|a_{1k}z_1^k|e_1+|a_{2k}z_2^k|e_2\nonumber\\
    &=\sum_{p=1}^2\left|\frac{\prod\limits_{i=1}^{m}\Gamma(\mu_{pi}+kM_{pi})}{\prod\limits_{j=1}^{n}\Gamma(\nu_{pj}+kN_{pj})\Gamma(k+1)}\right||z_p|^ke_p\nonumber\\&\sim\sum_{p=1}^2\left|\frac{\prod\limits_{i=1}^{m}\sqrt{2\pi} \left(\mu_{pi}+kM_{pi}\right)^{\mu_{pi}+kM_{pi}-\frac{1}{2}} e^{-(\mu_{pi}+kM_{pi})}}{\prod\limits_{j=1}^{n}\sqrt{2\pi} \left(\nu_{pj}+kN_{pj}\right)^{\nu_{pj}+kN_{pj}-\frac{1}{2}} e^{-(\nu_{pj}+kN_{pj})}\times\sqrt{2\pi}\;(k+1)^{k+\frac{1}{2}}e^{-k-1}}\right|\mathcal{V}_p^ke_p\nonumber\\
    &=(2\pi)^{\frac{m-n-1}{2}}\sum_{p=1}^2\left[\frac{\prod\limits_{i=1}^{m} \left(k|M_{pi}|\right)^{\mu_{pi}+kM_{pi}-\frac{1}{2}} e^{-(\mu_{pi}+kM_{pi})}}{\prod\limits_{j=1}^{n} \left(k|N_{pj}|\right)^{\nu_{pj}+kN_{pj}-\frac{1}{2}} e^{-(\nu_{pj}+kN_{pj})}\times\;k^{k+\frac{1}{2}}e^{-k-1}}\right]\mathcal{V}_p^ke_p\nonumber\\
    &=(2\pi)^{\frac{m-n-1}{2}}\sum_{p=1}^2\frac{e^{[\lambda_p+\frac{n-m}{2}+k(\Upsilon_p+1)]}}{k^{[\lambda_p+k(\Upsilon_p+1)+\frac{1}{2}]}}\frac{(|M_{pi}|)^{\mu_{pi}-\frac{1}{2}}}{(|N_{pj}|)^{\nu_{pj}-\frac{1}{2}}}e_p\nonumber\\
    &=(2\pi)^{\frac{m-n-1}{2}}\sum_{p=1}^2\frac{e^{[\lambda_p+\frac{n-m}{2}]}}{k^{[\lambda_p+\frac{1}{2}]}}\frac{(|M_{pi}|)^{\mu_{pi}-\frac{1}{2}}}{(|N_{pj}|)^{\nu_{pj}-\frac{1}{2}}}e_p.
    \end{align}
    Now let us take $\mathcal{M}_k=(2\pi)^{\frac{m-n-1}{2}}\sum\limits_{p=1}^2\frac{e^{[\lambda_p+\frac{n-m}{2}]}}{k^{[\lambda_p+\frac{1}{2}]}}\frac{(|M_{pi}|)^{\mu_{pi}-\frac{1}{2}}}{(|N_{pj}|)^{\nu_{pj}-\frac{1}{2}}}e_p$. 
    Using \eqref{eq:x4} and \eqref{eq:s6}, it follows that the series $\sum \mathcal{M}_k$ is hyperbolic convergent. Consequently, by the hyperbolic Weierstrass M-test \cite{bc_root}, the series $\sum\limits_{k=0}^{\infty}a_kZ^k$ is absolutely hyperbolic convergent. Therefore, the proof is established.
\end{proof}
\begin{remark}
    By setting $M_i=N_j=1$ for $i=1,2,\ldots m$ and $j=1,2,\ldots n$ in \eqref{eq:fw}, we obtain bicomplex Fox-Wright function represented in terms of bicomplex generalized hypergeometric function as follows: 
    \begin{align}\label{eq:d6}
        {}_{m}\Psi_{n}\left[\begin{matrix}&(\mu_1,1),\ldots &(\mu_m,1);\\&(\nu_1,1),\ldots&(\nu_n,1);&\end{matrix}Z\right]=
  \sum_{k=0}^{\infty}\frac{\prod\limits_{i=1}^{m}\Gamma_b(\mu_i+k)}{\prod\limits_{j=1}^{n}\Gamma_b(\nu_j+k)}\frac{Z^k}{k!}
  =\frac{\prod\limits_{i=1}^{m}\Gamma_b(\mu_i)}{\prod\limits_{j=1}^{n}\Gamma_b(\nu_j)}{}_mF_n\left[\begin{matrix}&\mu_1,\ldots &\mu_m;\\&\nu_1,\ldots&\nu_n;&\end{matrix}Z\right].
    \end{align}
    In \rm{\cite{bicomplex ghf}}, the authors examined the convergence of the bicomplex generalized hypergeometric function ${}_mF_n\left[\begin{matrix}&\mu_1,\ldots &\mu_m;\\&\nu_1,\ldots&\nu_n;&\end{matrix}Z\right]$. By choosing $M_i=N_j=1$ for $i=1,2,\ldots m$ and $j=1,2,\ldots n$ in Theorem \ref{th:a5}, we obtain a convergence condition for the bicomplex generalized hypergeometric function that agrees with our findings.
\end{remark}
\begin{remark}
    Setting $m=n=\mu_1=M_1=1$ in \eqref{eq:fw}, we obtain bicomplex two-parameter Mittag-Leffler function is given by
    \begin{align}\label{eq:b9}
        {}_{1}\Psi_{1}\left[\begin{matrix}&(1,1);\\&(\nu_1,N_1);&\end{matrix}Z\right]=\sum_{k=0}^\infty \frac{Z^k}{\Gamma_b(\nu_1+kN_1)}=E_{N_1,\nu_1}(Z),
    \end{align}
   which was introduced in \rm{\cite{bic_two_mittag}}.
\end{remark}
\begin{remark}
    Setting $i=0,j=1,\nu_1=\mathcal{V}+1,N_1=1$ and $Z=-\frac{Y^2}{4}$ in \eqref{eq:fw}, we have 
     \begin{align}
        {}_{0}\Psi_{1}\left[\begin{matrix}&-;\\&(\mathcal{V}+1,1);&\end{matrix}-\frac{Y^2}{4}\right]=\sum_{k=0}^\infty \frac{Y^{2k}}{4^k\Gamma_b(\mathcal{V}+1+k)k!}=\left(\frac{Y}{2}\right)^{-\mathcal{V}}J_{\mathcal{V}}(Y),
    \end{align}
    where $J_{\mathcal{V}}(Y)$ denotes the bicomplex bessel function. 
\end{remark}
\begin{corollary}
    From the first case of Theorem \ref{th:a5}, if $\Upsilon>_h -1$ then both components of the hyperbolic radius of convergence $R$ are infinite. Hence, according to Theorem $6$ of \rm{\cite{bc_root}}, the bicomplex Fox-Wright function is an entire function in $\mathbb{BC}$. 
\end{corollary}
\section{Bicomplex Fox-Wright coherent states}\label{sec5}
 This section develops coherent states related to the bicomplex Fox–Wright function and, via a discrete-to-continuous limit, obtains those corresponding to the continuous spectrum. A bicomplex FW-generalized multi-parameter $\nu$-function is also introduced and show that it acts as the normalization function for these bicomplex Fox–Wright coherent states in the continuous spectrum. In \cite{bicomplex ghf}, we analyzed coherent states related to the bicomplex generalized hypergeometric function. Here, we turn our attention to bicomplex Fox–Wright coherent states. In an infinite-dimensional bicomplex Hilbert space the Fock states are defined by \cite{bicomplex ghf}:
\begin{equation}\label{fock}
\mid k>=\sum_{p=1}^{2}\mid k_p> e_p,\quad k_p=0,1,2,...\nonumber
\end{equation}
and these states satisfy the corresponding orthogonality and completeness relations, which are given by:
\begin{equation}\label{csn3}
<k\mid l>=\sum_{p=1}^{2}\delta_{k_p l_p}e_p\;\;\mbox{and}\;\; \sum_{k=0}^{\infty}\mid k><k\mid=1.\nonumber
\end{equation}
Following the same approach as in the complex case, we introduce the bicomplex Fox–Wright states corresponding to the discrete spectrum, expressed in terms of the Fock states $\mid k>$ as follows
\begin{eqnarray}\label{cs1}
\mid m;n;Z> =\frac{1}{\sqrt{\mathcal{N}^{(m,n)}(\mid Z\mid_h^2)}}\sum_{k=0}^\infty \frac{Z^k}{\sqrt{\rho^{(m,n)}(k)}}\mid k>,
\end{eqnarray}
in which the normalization functions are represented by the bicomplex Fox–Wright functions in the following form
\begin{align}\label{cs3}
  \mathcal{N}_b^{(m,n)}(W)&= \frac{\prod\limits_{j=1}^n \Gamma_b(\nu_j)}{\prod\limits_{i=1}^m\Gamma_b(\mu_i)} {}_{m}\Psi_{n}\left[\begin{matrix}&(\mu_1,M_1),\ldots &(\mu_m,M_m);\\&(\nu_1,N_1),\ldots&(\nu_n,N_n);&\end{matrix}W\right]\nonumber\\
&=\frac{\prod\limits_{j=1}^n \Gamma_b(\nu_j)}{\prod\limits_{i=1}^m\Gamma_b(\mu_i)}\sum_{p=1}^{2}\left[\sum_{k_p=0}^\infty\frac{\prod\limits_{i=1}^m\Gamma(\mu_{pi}+k_pM_{pi})}{\prod\limits_{j=1}^n\Gamma(\nu_{pj}+k_pN_{pj})}\frac{w_p^{k_p}}{k_p!}\right]e_p
\end{align}
 for $W=w_1e_1+w_2e_2=\mid Z \mid_h^2$.
The parameter function $\rho_b^{(m,n)}(k)$ is given by
\begin{eqnarray}\label{cs2}
\rho_b^{(m,n)}(k) =\sum_{p=1}^2 \Gamma(k_p+1)\frac{\prod\limits_{i=1}^m\Gamma(\mu_{pi})}{\prod\limits_{j=1}^n \Gamma(\nu_{pj})}\frac{\prod\limits_{j=1}^n\Gamma(\nu_{pj}+k_pN_{pj})}{\prod\limits_{i=1}^m\Gamma(\mu_{pi}+k_pM_{pi})}e_p.
\end{eqnarray}
Let us consider $r=r_1e_1+r_2e_2$ and similarly define a bicomplex-valued function
\begin{equation*}\label{eq:c3s}
f_b^{(m,n)}(r)=\sum_{p=1}^2\sqrt{\frac{\prod\limits_{i=1}^m\Gamma(\mu_{pi}+r_pM_{pi})\prod\limits_{j=1}^n\Gamma(\nu_{pj}+(r_p+1)N_{pj})}{\prod\limits_{i=1}^m\Gamma(\mu_{pi}+(r_p+1)M_{pi})\prod\limits_{j=1}^n\Gamma(\nu_{pj}+r_pN_{pj})}(r_p+1)}\;e_p,
\end{equation*}
then parameter function $\rho_b^{(m,n)}(k)$ satisfies the following recurrence relations
\begin{equation}\label{cs2b}
\rho_b^{(m,n)}(0)=\sum\limits_{s=1}^2 e_s=1\;\; \mbox{and}\;\;\rho_b^{(m,n)}(k+1)=\rho_b^{(m,n)}(k)\left(f_b^{(m,n)}(k)\right)^2.
\end{equation}
 It follows from the preceding construction that $\rho^{(m,n)}(k)$ must be a strictly positive hyperbolic number, thereby imposing conditions on the associated parameters
\begin{equation*}
\mu_{pi}>0,\nu_{pj}>0\;\quad \mbox{for}\;i=1,2,...,m; j=1,2,...,n; p=1,2.
\end{equation*}
Analogous to the complex case, by employing the normalization functions, we derive the scalar product between two bicomplex Fox–Wright states as given below.
\begin{equation}\label{cs4}
<m;n;Z \mid m;n;Z'>=\frac{\mathcal{N}_b^{(m,n)}(Z*Z')}{\sqrt{\mathcal{N}_b^{(m,n)}(\mid Z\mid_h^2)}\sqrt{\mathcal{N}_b^{(m,n)}(\mid Z'\mid_h^2)}}.
\end{equation}
Equation \eqref{cs4}, shows that the bicomplex Fox–Wright states are normalized but non-orthogonal, and is well defined when the involved bicomplex Fox–Wright functions are convergent. Theorem \ref{th:a5} provides the convergence conditions for the bicomplex Fox–Wright function and we consider the first condition $\Upsilon>_h-1$.
Now, we define annihilation and creation operators as follows
\begin{eqnarray}
\mathcal{U}_-&=&\sum_{k=0}^\infty f_b^{(m,n)}(k)\mid k><k+1\mid\label{csn1}\nonumber\\
\mathcal{U}_+&=&\sum_{k=0}^\infty f_b^{(m,n)}(k)\mid k+1><k\mid\label{csn2}\nonumber
\end{eqnarray}
and similarly we obtain the following relations:
\begin{align}\label{csn5}
&\mathcal{U}_-\mid k>=f_b^{(m,n)}(k-1)\mid k-1>,\nonumber\\
&\mathcal{U}_+\mid k>=f_b^{(m,n)}(k)\mid k+1>,\\
&<k\mid \mathcal{U}_-\mathcal{U}_+\mid k>=\left[f_b^{(m,n)}(k)\right]^2,\nonumber
\end{align}
\begin{equation}\label{csf6}
<k\mid \mathcal{U}_+\mathcal{U}_-\mid k>=\left[f_b^{(m,n)}(k-1)\right]^2,
\end{equation}
and
\begin{equation}\label{csn9}
\left[\;\mathcal{U}_-,\mathcal{U}_+\right]=\sum_{k=0}^{\infty}\left(\left[f_b^{(m,n)}(k)\right]^2-\left[f_b^{(m,n)}(k-1)\right]^2\right)\mid k><k\mid.\nonumber
\end{equation}
Applying the recurrence relations \eqref{cs2b} and following the same procedure as in the complex case, one can readily verify that $\mathcal{U}_-\mid m;n;Z>=Z\mid m;n;Z>$. It follows that the bicomplex Fox–Wright states $\mid m;n;Z>$ are eigenstates of the annihilation operator $\mathcal{U}_-$ corresponding to the bicomplex eigenvalue $Z$ and thus constitute coherent states.\\
 Let us take the ground state $\mid 0>=\mid 0>e_1+\mid 0>e_2$, for which the annihilation operator acts as $\mathcal{U}_- \mid 0>=0\mid 0>e_1+0\mid 0>e_2$. Now, by successively applying the $k$-fold creation operator $\mathcal{U}_+$ to the ground $\mid 0>$ and using \eqref{csn5}, we obtain
\begin{align*}
    (\mathcal{U}_+)^k\mid 0>&=\prod_{r=1}^{k}f^{(m,n)}_b(r)\sum_{p=1}^2\mid k_p>e_p\\
    &=\sqrt{\rho_b^{(m,n)}(k)}\mid k>.
\end{align*}
As $\mathcal{U}_+$ is the adjoint operator of $\mathcal{U}_-$, similar to the complex case, we derive the following relations. 
\begin{align}\label{eq:csb}
    \mid k>=\frac{1}{\sqrt{\rho_b^{(m,n)}(k)}}(\mathcal{U}_+)^k\mid 0>\;\; \mbox{and} \;\;< k \mid=\frac{1}{\sqrt{\rho_b^{(m,n)}(k)}} <0 \mid (\mathcal{U}_-)^k.
\end{align}
Using \eqref{eq:csb} and \eqref{cs1}, we get
\begin{align}\label{eq:csn}
    \mid m;n;Z>= \frac{1}{\sqrt{\mathcal{N}_b^{(m,n)}(\mid Z\mid_h^2)}}\sum_{k=0}^\infty \frac{(Z\mathcal{U}_+)^k}{\rho_b^{(m,n)}(k)}\mid 0>
    \end{align}
    and
    \begin{align}\label{eq:c9sn}
    <m;n;Z \mid= \frac{1}{\sqrt{\mathcal{N}_b^{(m,n)}(\mid Z\mid_h^2)}}<0 \mid\sum_{k=0}^\infty \frac{(Z^\ast\mathcal{U}_-)^k}{{\rho_b^{(m,n)}(k)}}.
\end{align}
D. Popov et al. \cite{coh_cont}, developed coherent states for continuous spectrum via a discrete-to-continuous limit, we adopt the same technique in the bicomplex framework. Let us consider a dimensionless bicomplex Hamiltonian $H$ with non-degenerate continuous spectrum and we define dimensionless eigen states $\mid E>$ as follows:
\begin{align}
    \mid E>=\sum_{p=1}^2\mid E_p>e_p,\;\; 0\leq E_p\leq \infty
\end{align}
with
\begin{align*}
   <E\mid F>=\sum\limits_{p=1}^2\delta(E_p-F_p)e_p.
\end{align*}
The transition from discrete spectrum to continuous spectrum (limit $d\rightarrow c$) is given by $k_p$ replaced by $E_p$ for $p=1,2$, sum replaced by integrals and Kronecker deltas become Dirac deltas:
\begin{align*}
    \sum_{k_p=0}^\infty h(k_p)\rightarrow \int_{0}^\infty dE_p{\Tilde{h}(E_p)}, \delta_{k_p l_p}\rightarrow \delta(E_p-F_p).
\end{align*}
The completeness relation of the eigen state $\mid E>$ for continuous spectrum is define as follows:
\begin{align}\label{cs56}
    \int\limits_{0}^\infty dE\mid E><E\mid= \sum\limits_{p=1}^2\int\limits_{0}^\infty dE_p\mid E_p><E_p\mid e_p=I.
\end{align}
First, we derived bicomplex Fox-Wright coherent states given in \eqref{cs1} for continuous spectrum and subsequently, we showed that these coherent states for continuous spectrum fulfill the following properties:
\begin{itemize}
    \item [(i)] the set $\mid \Tilde{Z}>$ are normalizable i,e. $<\Tilde{Z} \mid \Tilde{Z}>=1$,
    \item [(ii)] the set $\mid \Tilde{Z}>$ are continuous i,e. $|Z'-Z|_h\rightarrow 0$ implies that $||\mid \Tilde{Z}'>-\mid \Tilde{Z}>||\rightarrow 0$,
    \item[(iii)] the set $\mid \Tilde{Z}>$ satisfies resolution of unity i,e. a weight function $\Tilde{W}_b(|Z|^2_h)$ is chosen of the integration measure $$d\Tilde{\mu}_b(Z)=\frac{d\theta_b}{2\pi}d(|Z|^2_h)\Tilde{W}_b(|Z|^2_h)$$ such that the integral $\int d \Tilde{\mu}_b(Z)\mid \Tilde{Z}><\Tilde{Z}\mid=1$ holds, where $\theta_b=\theta_1e_1+\theta_2e_2,\;0<\theta_1,\theta_2<2\pi$. 
\end{itemize}

Apply discrete continuous limit $d\rightarrow c$ in \eqref{cs3} and \eqref{cs2}, we obtain
\begin{align}\label{eq:c2}
    \mathcal{\Tilde{N}}_b^{(m,n)}(W)=\lim_{d\rightarrow c}\mathcal{N}_b^{(m,n)}(W)&=\frac{\prod\limits_{j=1}^n \Gamma_b(\nu_j)}{\prod\limits_{i=1}^m\Gamma_b(\mu_i)}\sum_{p=1}^{2}\left[\int\limits_{0}^\infty dE_P\frac{\prod\limits_{i=1}^m\Gamma(\mu_{pi}+E_pM_{pi})}{\prod\limits_{j=1}^n\Gamma(\nu_{pj}+E_pN_{pj})}\frac{w_p^{E_p}}{\Gamma(E_p+1)}\right]e_p\nonumber\\
    &=\frac{\prod\limits_{j=1}^n \Gamma_b(\nu_j)}{\prod\limits_{i=1}^m\Gamma_b(\mu_i)}\int\limits_{0}^\infty dE\frac{\prod\limits_{i=1}^m\Gamma_b(\mu_{i}+EM_{i})}{\prod\limits_{j=1}^n\Gamma_b(\nu_{j}+EN_{j})}\frac{W^E}{\Gamma_b(E+1)}
\end{align}
and 
\begin{align}\label{eq:csr}
    \Tilde{\rho}_b^{(m,n)}(E)=\lim_{d\rightarrow c}\rho_b^{(m,n)}(k)=\sum_{p=1}^2 \Gamma(E_p+1)\frac{\prod\limits_{i=1}^m\Gamma(\mu_{pi})}{\prod\limits_{j=1}^n \Gamma(\nu_{pj})}\frac{\prod\limits_{j=1}^n\Gamma(\nu_{pj}+E_pN_{pj})}{\prod\limits_{i=1}^m\Gamma(\mu_{pi}+E_pM_{pi})}e_p.
\end{align}
It follows from \eqref{eq:c2}, that the two idempotent components of $\mathcal{\Tilde{N}}_b^{(m,n)}(W)$ correspond to FW-generalized multi-parameter $\nu$-functions. Consequently, this represents a bicomplex extension of the FW-generalized multi-parameter $\nu$-function. Then we refer to it as the bicomplex FW-generalized multi-parameter $\nu$-function and denote it by $\mathcal{V}^{(m,n)}_b(W)$. Therefore $\mathcal{\Tilde{N}}_b^{(m,n)}(W)=\mathcal{V}^{(m,n)}_b(W)$. Similarly, apply discrete continuous limit $d\rightarrow c$ in \eqref{cs1}, we have
\begin{align}
    \mid m;n;\Tilde{Z}>=\lim_{d\rightarrow c}\mid m;n;Z> =\frac{1}{\sqrt{\mathcal{V}^{(m,n)}_b(|Z|^2_h)}}\int\limits_{0}^\infty dE \frac{Z^E}{\sqrt{\Tilde{\rho}_b^{(m,n)}(E)}}\mid E>,
\end{align}
which is the expression of coherent states for continuous spectrum. The overlap of two coherent states for continuous spectrum is given by
\begin{equation}
<m;n;\Tilde{Z} \mid m;n;\Tilde{Z}'>=\lim_{d\rightarrow c}<m;n;Z \mid m;n;Z'>=\frac{\mathcal{V}_b^{(m,n)}(Z*Z')}{\sqrt{\mathcal{V}_b^{(m,n)}(\mid Z\mid_h^2)}\sqrt{\mathcal{V}_b^{(m,n)}(\mid Z'\mid_h^2)}},\nonumber
\end{equation}
which implies the bicomplex Fox-Wright coherent states for continuous spectrum are normalized.
Again apply discrete continuous limit $d\rightarrow c$ in \eqref{eq:csb}, we obtain
\begin{align}\label{cs23}
    \mid E>=\frac{1}{\sqrt{\Tilde{\rho}_b^{(m,n)}(E)}}(\mathcal{U}_+)^E\mid 0>\;\; \mbox{and} \;\;< E \mid=\frac{1}{\sqrt{\Tilde{\rho}_b^{(m,n)}(E)}} <0 \mid (\mathcal{U}_-)^E.
\end{align}
By applying equations \eqref{cs56} and \eqref{cs23} together with the DOOT rules \cite{doot}, we obtain the expression for the projector $\mid 0><0 \mid$ of the vacuum state $\mid 0>$ associated with the continuous spectrum
\begin{align}
   I=\int\limits_{0}^\infty dE \mid E><E\mid &=\mid 0><0 \mid \int\limits_{0}^\infty dE\frac{1}{\Tilde{\rho}_b^{(m,n)}(E)} \#(\mathcal{U}_+ \mathcal{U}_-)^E\#\nonumber\\
    &=\mid 0><0 \mid \sum_{p=1}^{2}\left[\int\limits_{0}^\infty dE_P\frac{\prod\limits_{j=1}^n \Gamma(\nu_{pj})\prod\limits_{i=1}^m\Gamma(\mu_{pi}+E_pM_{pi})}{\prod\limits_{i=1}^m\Gamma(\mu_{pi})\prod\limits_{j=1}^n\Gamma(\nu_{pj}+E_pN_{pj})}\frac{\#(\mathcal{U}_+ \mathcal{U}_-)^{E_p}\#}{\Gamma(E_p+1)}\right]e_p\nonumber\\
    &=\mid 0><0 \mid \#\mathcal{V}^{(m,n)}_b(\mathcal{U}_+ \mathcal{U}_-)\#,\nonumber
\end{align}
which implies that
\begin{align}
    \mid 0><0 \mid=\frac{1}{\#\mathcal{V}^{(m,n)}_b(\mathcal{U}_+ \mathcal{U}_-)\#}.
\end{align}
By applying the discrete-continuous limit in \eqref{eq:csn} and \eqref{eq:c9sn}, we get
\begin{align}\label{eq:csi}
    \mid m;n; \Tilde{Z}>= \frac{1}{\sqrt{\mathcal{V}^{(m,n)}_b(\mid Z\mid_h^2)}}\int\limits_{0}^\infty dE\frac{(Z\mathcal{U}_+)^E}{\Tilde{\rho}_b^{(m,n)}(E)}\mid 0>
    \end{align}
    and
    \begin{align}\label{eq:csii}
    <m;n;\Tilde{Z} \mid= \frac{1}{\sqrt{\mathcal{V}^{(m,n)}_b(\mid Z\mid_h^2)}}<0 \mid\int\limits_{0}^\infty dE\frac{(Z^\ast\mathcal{U}_-)^E}{{\Tilde{\rho}_b^{(m,n)}(E)}}.
\end{align}
Next, our goal is to determine the weight function $\Tilde{W}(|Z|^2_h)$ so that bicomplex generalized coherent state fulfill the resolution of unity
\begin{align}\label{eq:csm}
    &\int d \Tilde{\mu}_b(Z)\mid m;n;\Tilde{Z}><m;n;\Tilde{Z}\mid=1.
\end{align}
Using \eqref{eq:csi},\eqref{eq:csii} and \eqref{eq:csm}, we get
\begin{align}\label{eq:csp}
    & \int \frac{d\theta_b}{2\pi}d(|Z|^2_h)\Tilde{W}(|Z|^2_h)\frac{\mid 0><0\mid}{\mathcal{V}_b^{(m,n)}(\mid Z\mid_h^2)}\int\limits_{0}^\infty dE\frac{(|Z|_h^2)^E}{{[\Tilde{\rho}_b^{(m,n)}(E)]^2}}\#(\mathcal{U}_+\mathcal{U}_-)^E\#=1\nonumber\\
    &\Rightarrow  \left[\sum_{p=0}^2\int_{0}^{2\pi} \frac{d\theta_p}{2\pi}e_p\right]\left[\int_{0}^{\infty}d(|Z|^2_h)\frac{\Tilde{W}(|Z|^2_h)\mid 0><0\mid}{\mathcal{V}_b^{(m,n)}\left(\mid Z\mid_h^2\right)}\int\limits_{0}^\infty dE\frac{(|Z|_h^2)^E}{{[\Tilde{\rho}_b^{(m,n)}(E)]^2}}\#(\mathcal{U}_+\mathcal{U}_-)^E\#\right]=1\nonumber\\
    &\Rightarrow \frac{1}{\#\mathcal{V}^{(m,n)}_b(\mathcal{U}_+ \mathcal{U}_-)\#}\left[\int\limits_{0}^\infty dE\frac{\#(\mathcal{U}_+\mathcal{U}_-)^E\#}{{[\Tilde{\rho}_b^{(m,n)}(E)]^2}}\int_{0}^{\infty}d(|Z|^2_h)\frac{\Tilde{W}(|Z|^2_h)(|Z|_h^2)^E}{\mathcal{V}^{(m,n)}_b\left(\mid Z\mid_h^2\right)}\right]=1
    \end{align}
    Let us take $\frac{\Tilde{W}(|Z|^2_h)}{\mathcal{V}^{(m,n)}_b\left(\mid Z\mid_h^2\right)}=\sum\limits_{p=1}^2h_{p}e_p$, substitute in \eqref{eq:csp} and using \eqref{eq:csr}, we get
    \begin{align}\label{eq:cp4}
    & \sum_{p=0}^2 \int\limits_{0}^\infty dE_p\frac{\prod\limits_{j=1}^n \Gamma(\nu_{pj})\prod\limits_{i=1}^m\Gamma(\mu_{pi}+E_pM_{pi})}{\prod\limits_{i=1}^m\Gamma(\mu_{pi})\prod\limits_{j=1}^n\Gamma(\nu_{pj}+E_pN_{pj})}\frac{\#(\mathcal{U}_+\mathcal{U}_-)^{E_p}\#}{\Gamma(E_p+1)}\nonumber\\&\hspace{30pt}\left[\frac{\prod\limits_{j=1}^n \Gamma(\nu_{pj})\prod\limits_{i=1}^m\Gamma(\mu_{pi}+E_pM_{pi})}{\prod\limits_{i=1}^m\Gamma(\mu_{pi})\prod\limits_{j=1}^n\Gamma(\nu_{pj}+E_pN_{pj})}\frac{1}{\Gamma(E_p+1)}\int\limits_{0}^\infty dw_ph_p(w_p)^{E_p}\right]e_p=\#\mathcal{V}^{(m,n)}_b(\mathcal{U}_+ \mathcal{U}_-)\#.
    \end{align}
    Using \eqref{cs3} and \eqref{eq:cp4}, it is found that the value of the second integral is
    \begin{align}\label{eq:cha}
        \int\limits_{0}^\infty dw_ph_p(w_p)^{E_p}=\Gamma(E_p+1)\frac{\prod\limits_{i=1}^m\Gamma(\mu_{pi})}{\prod\limits_{j=1}^n \Gamma(\nu_{pj})}\frac{\prod\limits_{j=1}^n\Gamma(\nu_{pj}+E_pN_{pj})}{\prod\limits_{i=1}^m\Gamma(\mu_{pi}+E_pM_{pi})}.
    \end{align}
  Substituting 
$E_p=s-1$ into equation \eqref{eq:cha} and using the Mellin transform of a $H$-function \eqref{eq:mlt7}, we obtain 
\begin{align*}
    h_p=\frac{\prod\limits_{i=1}^m\Gamma(\mu_{pi})}{\prod\limits_{j=1}^n \Gamma(\nu_{pj})}
    H_{m,n+1}^{n+1,0} \left[ |z_p|^2 \;\middle|\; \begin{matrix} &(\mu_{p1}-M_{p1},M_{p1}),&(\mu_{p2}-M_{p2},M_{p2}) \ldots, &(\mu_{pm}-M_{pm},M_{pm}) \\ &(0,1),&(\nu_{p1}-N_{p1},N_{p1}), \ldots, &(\nu_{pn}-N_{pn},N_{pn}) \end{matrix} \right],
\end{align*}
 for $p=1,2$ and the weight function is given by
\begin{align*}
   &\Tilde{W}(|Z|^2_h)=\mathcal{V}^{(m,n)}_b\left(\mid Z\mid_h^2\right)\\&\hspace{10pt}\sum_{p=1}^2\frac{\prod\limits_{i=1}^m\Gamma(\mu_{pi})}{\prod\limits_{j=1}^n \Gamma(\nu_{pj})}
    H_{m,n+1}^{n+1,0} \left[ |z_p|^2 \;\middle|\; \begin{matrix} &(\mu_{p1}-M_{p1},M_{p1}),&(\mu_{p2}-M_{p2},M_{p2}) \ldots, &(\mu_{pm}-M_{pm},M_{pm}) \\ &(0,1),&(\nu_{p1}-N_{p1},N_{p1}), \ldots, &(\nu_{pn}-N_{pn},N_{pn}) \end{matrix} \right]e_p.
\end{align*}
Hence, the corresponding integration measure for continuous spectrum is given by
\begin{align*}
    &d\Tilde{\mu}_b(Z)=\frac{d\theta_b}{2\pi}d(|Z|^2_h)\mathcal{V}^{(m,n)}_b\left(\mid Z\mid_h^2\right)\\&\hspace{10pt}
\sum_{p=1}^2\frac{\prod\limits_{i=1}^m\Gamma(\mu_{pi})}{\prod\limits_{j=1}^n \Gamma(\nu_{pj})}
    H_{m,n+1}^{n+1,0} \left[ |z_p|^2 \;\middle|\; \begin{matrix} &(\mu_{p1}-M_{p1},M_{p1}),&(\mu_{p2}-M_{p2},M_{p2}) \ldots, &(\mu_{pm}-M_{pm},M_{pm}) \\ &(0,1),&(\nu_{p1}-N_{p1},N_{p1}), \ldots, &(\nu_{pn}-N_{pn},N_{pn}) \end{matrix} \right]e_p.
\end{align*}
Thus the bicomplex Fox-Wright coherent states given in \eqref{cs1} satisfies the fundamental properties of normalizability, continuity, and resolution of unity. Now, if we consider the inverse process, that is the continuous–discrete limit $c\rightarrow d$, we obtain integration measure for discrete spectrum which is given by
\begin{align}\label{eq:c4sf}
    d\mu_b(Z)&=\lim_{c\rightarrow d}d\Tilde{\mu}_b(Z)\nonumber\\&=\frac{d\theta_b}{2\pi}d(|Z|^2_h){}_{m}\Psi_{n}\left[\begin{matrix}&(\mu_1,M_1),\ldots &(\mu_m,M_m);\\&(\nu_1,N_1),\ldots&(\nu_n,N_n);&\end{matrix}\mid Z\mid_h^2\right]\nonumber\\&\hspace{20pt}
\sum_{p=1}^2 H_{m,n+1}^{n+1,0} \left[ |z_p|^2 \;\middle|\; \begin{matrix} &(\mu_{p1}-M_{p1},M_{p1}),&(\mu_{p2}-M_{p2},M_{p2}) \ldots, &(\mu_{pm}-M_{pm},M_{pm}) \\ &(0,1),&(\nu_{p1}-N_{p1},N_{p1}), \ldots, &(\nu_{pn}-N_{pn},N_{pn}) \end{matrix} \right]e_p.
\end{align}
\begin{corollary}\label{cr1}
    If $M_i=N_j=1$ for $i=1,2,\ldots,m$ and $j=1,2,\ldots,n$ in \eqref{cs3},\eqref{cs2} and \eqref{eq:c4sf}, then we recover the bicomplex generalized hypergeometric coherent states $\mid m;n;Z>$ \rm{\cite{bicomplex ghf}}, with the normalization function expressed as 
    \begin{align*}
        \mathcal{N}_b^{(m,n)}(W)={}_{m}F_{n}\left[\begin{matrix}&\mu_1,\mu_2,\ldots &\mu_m;\\&\nu_1,\nu_2,\ldots&\nu_n;&\end{matrix}W\right],
    \end{align*}
   and corresponding integration measure for discrete spectrum is given by
   \begin{align*}
        &d\mu_b(Z)=\frac{d\theta_b}{2\pi}d(|Z|^2_h){}_{m}F_{n}\left[\begin{matrix}&\mu_1,\ldots &\mu_m;\\&\nu_1,\ldots&\nu_n;&\end{matrix}\mid Z\mid_h^2\right]\sum_{p=1}^2 G_{m,n+1}^{n+1,0} \left[ |z_p|^2 \;\middle|\; \begin{matrix} &\mu_{p1}-1,&\mu_{p2}-1 \ldots, &\mu_{pm}-1 \\ &0,&\nu_{p1}-1, \ldots, &\nu_{pn}-1 \end{matrix} \right]e_p,
   \end{align*}
   where $\theta_b=\theta_1e_1+\theta_2e_2,\;0<\theta_1,\theta_2<2\pi$.
\end{corollary}
\section{Concluding remarks}
In this paper, we have introduced a new class of coherent states in which the Fox–Wright function appears as the normalization function. Consequently, these coherent states can be called Fox-Wright coherent states. We have shown that these states satisfy the fundamental conditions of continuity, normalizability and resolution of unity. According to Remark \ref{rm1} by assigning specific values to the parameters $A_l$ and $B_r$ coherent states associated with Fox-Wright function coincides with the generalized hypergeometric coherent states presented in \cite{gh_cs}. In the same manner, particular cases lead to different classes of coherent states involving other special functions. In addition, we have defined the Fox–Wright function in bicomplex space and its existence has been examined across nine distinct cases. It has been shown that when $\Upsilon=-1$ the corresponding series is absolutely hyperbolically convergent inside the hyperbolic ball $\mathcal{B}_h(0,\mathcal{V})$ and on its boundary the convergence is uniform and absolute in the hyperbolic sense whenever $\Re(\Lambda_1)-\frac{1}{2}>|\Im(\Lambda_2)|$ . Later, we have derived an application of the bicomplex Fox–Wright function in the construction of coherent states, referred to as bicomplex Fox–Wright coherent states. Moreover, we have proposed a framework for transition from the discrete spectrum to the continuous spectrum in the bicomplex setting. Additionally, we have defined a bicomplex-valued function, namely the bicomplex FW-generalized multi-parameter $\nu$-function $\mathcal{V}^{(m,n)}_b(W)$ and discussed its connection with bicomplex Fox–Wright coherent states for the continuous spectrum, where it serves as the normalization function. Corollary \ref{cr1}, shows that, for appropriate choices of the parameters $M_i$ and $N_j$ the coherent states associated with the bicomplex Fox–Wright function reduce to the bicomplex generalized hypergeometric coherent states obtained in \rm{\cite{bicomplex ghf}}.\\\\
\textbf{Acknowledgment.} The author “Snehasis Bera” expresses sincere gratitude to the “University Grants Commission”, India for providing financial support (Benificiary Code: BWBME00403738) for carrying out this research work.

\end{document}